\documentclass[12pt]{amsart}
\usepackage{latexsym,fancyhdr,amssymb,color,amsmath,amsthm,graphicx,listings,comment}
\usepackage[section]{placeins}
\newtheorem{thm}{Theorem} \newtheorem{propo}{Proposition} 
\newtheorem{lemma}{Lemma}  \newtheorem{coro}{Corollary}
\definecolor{red1}{rgb}{1,0.9,0.9} \definecolor{blue1}{rgb}{0.9,0.9,1} \definecolor{green1}{rgb}{0.9,1,0.9} 
\definecolor{yellow1}{rgb}{1,1,0.9} \definecolor{yellow2}{rgb}{1,1,0.8}

\let\paragraph\subsection

\newcommand{\R}{\mathcal{R}}

\title{Listening to the cohomology of graphs}
\author{Oliver Knill}
\date{February 4, 2018}
\address{Department of Mathematics \\ Harvard University \\ Cambridge, MA, 02138 }

\begin{document}
\maketitle

\begin{abstract}
We prove that the spectrum of the Kirchhoff Laplacian $H_0$ of a finite simple Barycentric refined
graph and the spectrum of the unimodular connection Laplacian $L$ of $G$ determine each other.
Indeed, we note that $L-g$ with $g=L^{-1}$ is similar to the Hodge Laplacian 
$H=(d+d^*)^2$ of $G$ which is in one dimensions the direct sum 
$H= H_0 \oplus H_1 = d_0^* d_0 \oplus d_0 d_0^*$ 
of the Kirchhoff Laplacian $H_0$ and its 1-form analog $H_1$. 
In this one-dimensional case, the spectrum of a single choice of one of the 
three matrices $H_0, H_1$ or $H$ alone is enough to determine both Betti numbers 
$b_0,b_1$ of $G$ as well as the spectrum of the other matrices. 
It follows from the similarity of $H$ and $L-L^{-1}$ that for a one-dimensional complex 
which is a Barycentric refinement, the number of connectivity components $b_0$ is the number 
of eigenvalues $1$ of $L$ and that the genus $b_1$ is the number of eigenvalues $-1$ of $L$. 
It will also lead to much better estimates of spectral radius and algebraic connectivity.
For a general abstract finite simplicial complex $G$, we express the Green function values
$g(x,y) = \omega(x) \omega(y) \chi( {\rm St}(x) \cap {\rm St}(y) )$ in terms of the
stars ${\rm St}(x)= \{ z \in G | x \subset z \}$ of $x$ and 
$\omega(x)=(-1)^{{\rm dim}(x)}$. One can see $W^+(x)=\{ z \in G | x \subset z \}$ 
and $W^-(x)=\{ z \in G | z \subset x \}$ as stable and unstable manifolds of a simplex $x \in G$ and
$g(x,y)=\omega(x) \omega(y) \chi(W^+(x) \cap W^+(y))$ as heteroclinic intersection 
numbers or curvatures and the identity $L g = 1$ as a collection of Gauss-Bonnet formulas. 
A special case is the previously known $g(x,x)=\chi({\rm St}(x)) = 1-\chi(S(x))$.
The homoclinic energy
$\omega(x)=\chi(W^+(x) \cap W^-(x))$ by definition satisfies $\chi(G)=\sum_x \omega(x)$. The matrix
$M(x,y)=\omega(x) \omega(y) \chi(W^-(x) \cap W^-(y))$ which is similar to $L(x,y)=\chi(W^-(x) \cap W^-(y))$
has a sum of matrix entries which is Wu characteristic $\sum_{x \sim y} \omega(x) \omega(y)$. 
For $G$ with dimension $r \geq 2$ we don't know yet how to recover the Betti numbers 
$b_k$ from the eigenvalues of the matrix $H$ or from $L$. So far, it is only possible to get it from a 
collection of block matrices, via the Hodge relations $b_k = {\rm dim}(H_k)$. A natural conjecture is that if 
$G$ is a Barycentric refinement of an other complex, then the spectrum of $L$ determines the Betti 
vector $b$. This note only shows this to be true if $G$ has dimension $1$. 
\end{abstract}

\section{Introduction}

\paragraph{}
Any abstract finite simplicial complex $G$ defines a graph $\Gamma=(V,E)$, where 
$V=G$ and $E$ is the set of pairs of simplices $x,y \in G$, where either $x \subset y$ or
$y \subset x$. As the dimension ${\rm dim}: V \to R$ on $\Gamma$ is a coloring, the chromatic 
number of $\Gamma$ is equal to the clique number ${\rm dim}(G)+1$. The graph $\Gamma$ 
defines so a new complex $G_1$, the Whitney complex of $\Gamma$, which is called the 
Barycentric refinement of $G$. It consists of all subsets of $V$ in which all elements 
are connected to each other.  The connection matrix $L$ of $G_1$ is defined as 
$L(x,y)=1$ if $x \cap y \neq \emptyset$ and $L(x,y)=0$ else. 
If ${\rm dim}(G)=1$, then $G_1 = V \cup E$, with identifying $V$ with $\{ \{v\} | v \in V \}$. 
The connection matrix of $G_1$ is then an $n \times n$ matrix, where $n=|G_1|=|V|+|E|$. 

\paragraph{}
Given a finite abstract simplicial complex $G$ with $f$-vector $(v_0,v_1,$ $\dots,$ $v_r)$, 
its Dirac operator $D$ is $d+d^*$, where $d$ is the exterior derivative 
$d_k f(x_0, \dots, x_k) = \sum_{j} (-1)^j f(x_0, \dots, \hat{x}_j, \dots x_k)$
from $\Lambda_k \to \Lambda_{k+1}$, where $\Lambda_k$ is the linear
space of antisymmetric functions in $k+1$ variables defined on $k$-dimensional simplices. 
This space $\Lambda_k$ of discrete $k$-forms has dimension $v_k$, the cardinality of sets in $G$
of length $k+1$. Also the matrices $d_k$ or $D$ can then be realized as $n \times n$ matrices, where 
$n=\sum_{k=0}^r v_k$. Both $d$ and so $D$ depend on the orientation chosen for each simplex.
(The choice of orientation is a choice of basis in $\Lambda_d$ and has nothing to do with
orientability as no compatibility is required). 
The matrix $H=D^2$ is the Hodge Laplacian. It decomposes into blocks $H_k$, the individual 
form Laplacians. The kernel of $H_k$ has dimension $b_k$, which is the $k$'th Betti number. 
The vector $b=(b_0,b_1,\dots, b_r)$ is the Betti vector of $G$. 

\paragraph{}
There are two natural questions: can one read off the Betti vector
from the eigenvalues of $L$? Can one read off the Betti vector from 
the eigenvalues of $H$? We will see that in one dimensions, the two questions are
related but that there is already a subtlety: there are $L$-isospectral complexes
with different $b_0,b_1$ but that after a Barycentric refinement, we always can read off
$b_0,b_1$ from the eigenvalues of $L$. A weaker question is to get
the Euler characteristic $\chi(G) = \sum_k (-1)^k b_k$ from the eigenvalues, either in the 
$L$ or $H$ case. We have seen that one in general can get $\chi(G)$ from the eigenvalues of $L$ as $\chi(G)$ is 
the number of positive minus the number of negative eigenvalues \cite{HearingEulerCharacteristic}. 
In one dimensions, we can get $\chi(G)$ from the eigenvalues of $H$. We don't know how to get the Euler 
characteristic $\chi$ from the eigenvalues of $H$ in higher dimensions. 
We know by Hodge only $\sum_k b_k$, the nullity of $H$. 

\paragraph{}
For one-dimensional complexes, the single incidence matrix $d_0$ alone determines everything.
The matrix $d_0:\Lambda_0 \to \Lambda_1$ is the gradient, 
its transpose matrix $d_0^*: \Lambda_1 \to \Lambda_0$ is the divergence. The Hodge Laplacian 
$H=D^2$ has a block decomposition $H=H_0 \oplus H_1$ for 
$H_0=d_0^* d_0$ and $H_1 = d_0 d_0^*$. The matrix $H_0=B-A$ is the Kirchhoff matrix of
the graph $G$, where $B$ is the diagonal vertex degree matrix and $A$ is the 
adjacency matrix of $G$. The matrix $H_1=d_1 d_1^*$ is the Laplacian on $1$-forms. 
The $0$'th cohomology group $H^0(G)$ is defined as ${\rm ker}(d_0)$, the $1$'st cohomology 
group $H^1(G)$ is the linear space ${\rm ker}(d_1)/{\rm im}(d_0) = \R^m/{\rm im}(d)$ given by 
all function on all functions on edges modulo gradients. The nullity $b_0$ of $H_0$ is the 
number of connectivity components of $G$. 
The nullity $b_1$ of the matrix $H_1$ is the genus of $G$, which is the number of holes or 
the number of generators of the fundamental group.

\paragraph{}
The operators $L,g=L^{-1}$ and the operators $D,D^2=H$ are all defined on the 
same $n$-dimensional Hilbert space, where $n$ is the number of simplices in $G$. 
The matrices all depend on the order in which the simplices are arranged. The matrix
$D$ also depends on the orientation of simplicial complex. We can not only orient the 
one and higher dimensional simplices by assigning to them a preferred permutation,
allowing to define the orientation through the signature, we can also orient the zero
dimensional simplices. This changes $H$ too. In general, the matrix $H$ is always 
reducible, except if $|G|=1$.  On the other hand, the connection operator $L$ is irreducible 
if $G$ is connected. This means that there is a positive integer $k$ such that $L^k$ has only 
positive entries. Therefore, $L$ has a single Perron-Frobenius eigenvalue and this will be 
inherited by $H_0$. 

\paragraph{}
When estimating the largest and smallest non-zero eigenvalue of the Laplacian of a one-dimensional 
complex it does not matter whether we look for $H_0,H_1$ or $H$. Both the maximal eigenvalue = spectral radius $\rho$ 
as well as the second eigenvalue = ground state = algebraic connectivity $\alpha$ are of enormous interest. We plan to explore
this more elsewhere as in order to appreciate this fully, it requires to compare the improvement 
with the existing literature on the estimates but a preliminary assessment shows that the formula 
$H=L-L^{-1}$ is quite powerful to estimate both $\rho$ and $\alpha$
for Barycentric refinement graphs which, all examples of bipartite irregular graphs if $G$ is not circular.
(The omission of circular graphs is no issue as we know there all the eigenvalues explicitly). 
The estimates are often far better than the best established bounds. One of the questions
is to estimate how much below the largest eigenvalue has dropped below the obvious upper bound $2d$, where $d$ is the 
maximal vertex degree of the graph. We can use that $L$ is a non-negative matrix for which much is known 
\cite{MincNonnegative}.

\paragraph{}
The simplest spectral estimate of eigenvalues of $L$ is the upper bound $d_L+1$,
where $d_L$ is the maximal vertex degree of the line graph of $G$, the maximal spectral radius of $H$
is bound above by $d_L+1-1/(d_L+1)$ which is $\leq 2d-(1/2d)$ if $d$ is the maximal vertex degree.
This is better than estimates $2d-2/((2l+1) n)$ (\cite{Shi2007} Theorem 3.5) 
or $2d-1/(l n)$ (\cite{LiShiuChan2010} Theorem 2.3) established for irregular graphs of diameter 
$l$ and $n$ vertices. But our new estimate only holds after a Barycentric refinement.
But for higher Barycentric refinements we even have $\rho \leq d+2-1/(d+2)$ as the vertex degree of $L$
has then dropped considerably. The relation between $H$ and $L$ is so useful because $L$ is equal o $1+A$, 
where $A$ is the adjacency matrix of the connection graph of $G$ and $1$ is the identity matrix. 
The connection graph of $G$ is the graph in which two simplices are connected if they intersect. Any 
knowledge about the spectrum of adjacency matrices of graphs (like \cite{Stevanovic}) translates directly 
into spectral statements of $H$, whether it is spectral radius, algebraic connectivity or 
order structures of the eigenvalues. The reason is that the map $x \to x-1/x$ maps two spectral intervals 
of $L$ piecewise monotonically into the spectral interval of $H$.

\paragraph{}
Surprisingly little appears to be known about the relation of the Betti vector $b$ and the
eigenvalues $\sigma(H)$ of the Hodge Laplacian, both in the manifold as well as in the
case of simplicial complexes. In the graph case, a natural spectral problem is already to 
relate the spectrum of the adjacency matrix $A$ with the topology of the graph. Its spectrum determines
the number of edges and the number of triangles of the eigenvalues of $A$ \cite{Biggs}. 
For a $2$-dimensional connected complex for which every unit sphere is a $1$-dimensional circular 
graph $C_n$ with $n \geq 4$, we can then read off the Euler characteristic and so the
cohomology $\chi(G)=b_0-b_1+b_2=2-b_1$ of a connected oriented surface for which the boundary 
is a collection of closed circular graphs. For the Kirchhoff matrix $B-A$ of a graph,
we get the number of vertices, the trace is twice the number of edges from handshaking. 

\paragraph{}
We also can get
the Zagreb index $z(G)=\sum_{v \in V} {\rm deg}(v)^2$ from the spectrum and because $z(G)$ 
is related to the Wu characteristic $\omega(G)$ we also can get the Euler characteristic. 
(The connection of Wu characteristic with the Zagreb index has been pointed
out to us by Tamas Reti.) This analysis however requires the graph to be geometric, like being 
a nice triangulation of a $2$-dimensional surface with boundary. To illustrate how little is known, one can ask how
to read off the orientability of a triangulated surface from the eigenvalues
of the Hodge Laplacian $H$, without indication from which $k$-form vector 
the zero eigenvalues come from. While non-orientability implies a trivial kernel for the $2$-form Laplacian $H_2$, 
we don't yet know how to access non-orientability even in the two dimensional 
from the spectrum of $H=H_0 \oplus H_1 \oplus H_2$ or from the spectrum of $L$. 

\section{Relations between spectral data}

\paragraph{}
Let us report first a fact which is probably well known even-so we have not found
a reference despite the existence of a rather large literature on spectral graph theory. For books, see
\cite{Biggs,VerdiereGraphSpectra,DvetkovicDoobSachs,Brouwer,Chung97,Mieghem,Spielman2009,CRS}. 
As much of this literature is also about spectra of the adjacency matrix, the here discussed
relation between $L$ (a shifted adjacency matrix) and $H$ (a Laplacian matrix) is useful as
this relation is usually only available for regular graphs, where the vertex degree is constant. 

\paragraph{}
In the following, we mean with a graph $G=(V,E)$ the
$1$-dimensional simplicial complex $V \cup E$ defined by $G$ and not the Whitney complex.
For $G=K_3$ for example, this means $G=\{\{1\},\{2\}$, $\{3\},\{1,2\}$, $\{2,3\},\{3,1\} \}$ and not
the set of all subsets of $\{1,2,3\}$. In that case, $b_0,b_1$ are determined by the 
Kirchhoff Laplacian $H_0$ alone:

\begin{lemma}[Kirchhoff listens to the genus]
For any graph equipped with the 1-dimensional complex,
the eigenvalues of the Kirchhoff matrix $H_0$ determines the Betti numbers $b_0,b_1$
and the eigenvalues of $H$ as well as the eigenvalues of $H_1$. 
Two $H_0$ isospectral graphs are strongly isospectral, meaning $H$-isospectral. 
\end{lemma}
\begin{proof}
By Euler handshake, ${\rm tr}(H_0)=2|E|$. We also have ${\rm tr}(H_0^0)=|V|$. The number $b_0$
is the number of zero eigenvalues of $G$. Now, by Euler-Poincar\'e, 
$\chi(G)=b_0-b_1=|V|-|E|$. From this, one gets 
$$  b_1=b_0-|V|+|E|={\rm ker}(H_0) - {\rm tr}(1)-{\rm tr}(H_0)/2 $$
and the right hand side uses only spectral data. 
\end{proof} 

\paragraph{}
This can be compared with two-dimensional smooth regions in the plane
which can be glued together to build a Riemann surface, for which only $b_0$ and $b_1$
matter. The reason is that we essentially deal with a complex one-dimensional curve
then and that the double cover ramified over one-dimensional closed curves
is linked to the two dimensional region by Riemann-Hurwitz. So, it is no surprise that 
one can hear the genus of a drum \cite{Kac66}.
We are not aware of any other result, both in the 
simplicial complex, nor in the manifold case, where one can hear larger
Betti numbers $b_k(M)$ with $k \geq 2$ from the eigenvalues of any of the Laplacians 
on a three or higher dimensional manifold without looking at a sequence of form Laplacians 
$H_k$ which build up the Hodge Laplacian $H=\oplus_k H_k$. 

\paragraph{}
In the case of a $1$-dimensional simplicial complex, where are no higher exterior derivatives
like the curl $d_1$ have to be considered, the spectrum of $H_0$ determines completely
the spectrum of the Hodge operator $H_1$ and so $b_1$.

\begin{lemma}[Hodge listens to the genus] 
Let $H$ be the Hodge matrix of a $1$-dimensional simplicial complex $G$.
The eigenvalues of the matrix $H$ alone determine the Betti numbers $b_0,b_1$
of $G$. 
\end{lemma}
\begin{proof}
The matrix $H$ has two blocks $H_0=d_0^* d_0 = A^* A$ and $H_1=d_0 d_0^* = A A^*$. 
It is a general fact from linear algebra or the Cauchy-Binet formula for the
coefficients of the characteristic polynomial \cite{CauchyBinetKnill}
that $H_0$ and $H_1$ are essentially isospectral meaning that their non-zero
eigenvalues agree. It is also a very special case of McKean-Singer supersymmetry
which in general assures that the non-zero Bosonic and Fermonic spectra agree
for the Hodge Laplacian \cite{knillmckeansinger}. Now, from $H$, we can access 
${\rm ker}(H) = b_0+b_1$. We can also hear the number of eigenvalues as
${\rm tr}(1) = |V|+|E|$. The trace is ${\rm tr}(H) = 2|E|+2|E|=4|E|$. 
Therefore, we know both $|V|={\rm tr}(1)-{\rm tr}(H)/4$ and $|E|={\rm tr}(H)/4$ 
and so $\chi(G)=|V|-|E|=b_0-b_1$.
 Knowing $b_0+b_1$ and $b_0-b_1$ determines $b_0$ and $b_1$. 
\end{proof}

\paragraph{}
\begin{lemma}[Connection does not hear the genus]
There exist two one-dimensional simplicial complexes which are $L$-isospectral
but which have different $b_0,b_1$. 
\end{lemma}
\begin{proof}
The following pair of simplicial complexes was given in \cite{HearingEulerCharacteristic}.
The first one, $G$ is generated by the sets 
$$  \{\{1,2\},\{1,3\},\{2,6\},\{2,7\},\{6,8\},\{7,4\},\{4,5\} \} $$
The second one is generated by 
$$ H=\{\{1,2\},\{1,5\},\{1,7\},\{2,8\},\{5,6\},\{8,6\},\{3,4\}  \;. $$
\end{proof} 

\paragraph{}
We know however that in all dimensions, the eigenvalues of $L$ determine the Euler 
characteristic of $G$ as we have proven that in general, 
the Euler characteristic $\chi(G)$ is $p(G)-n(G)$,
where $p(G)$ and $n(G)$ are the number of positive and negative eigenvalues of 
the connection Laplacian $L=L(G)$ \cite{HearingEulerCharacteristic}. In the Barycentric
refined case, this will lead to a relation between the eigenvalues $1$ and $-1$ and 
the Betti numbers. The eigenvectors of $L-L^{-1}$ and $H$ are of course the same. 

\section{The theorem}

\paragraph{}
Our main result here relates the Hodge operator with the ``Hydrogen operator"
$L-L^{-1}$. The assumption of $G$ having chromatic number $2$ is not 
that severe but necessary as the above lemma shows. 
Every Barycentric refinement of a graph has chromatic number $2$, the color
being the dimension function. 

\begin{thm}[Hydrogen and Hodge]
Given a graph $\Gamma$ which is a Barycentric refinement of a one-dimensional
complex, then $L-L^{-1}$ and $H$ are similar. In a suitable basis: $H=L-L^{-1}$.
\end{thm}

We will prove this in the next section.
The etymology of "Hydrogren" was explained in \cite{DehnSommerville}, where we looked at the 
functional ${\rm tr}(L-L^{-1})$ which is in general of geometric interest as it is 
$\sum_x \chi(S(x))$. In $\R^3$ with Laplacian $L=-\Delta$, the kernel of the inverse 
$L^{-1}$ is the Newton potential $V_x(y)=1/(4 \pi |x-y|)$ because of Gauss 
${\rm div} {\rm grad}(1/|y-x|) = - 4\pi \delta(x)$. In quantum mechanics, the Hydrogen Hamiltonian 
is $-h^2/(2m) \Delta - e^2/(4 \pi e_0 r)$ has at least a formal analogy to $L-L^{-1}$.

\paragraph{}
Here is a simple example. We take $G=\{ \{1\},\{2\},\{3\},\{1,2\},\{2,3\} \}$
leading to a graph $\Gamma$ with $9$ vertices. It is the simplicial complex
$G=\{\{1\},\{2\},\{3\},\{4\},\{5\},\{1,4\},\{2,4\},\{2,5\},\{3,5\}\}$. 
First, lets write down the connection Laplacian:
$$ L=              \left[
                   \begin{array}{ccccccccc}
                    1 & 0 & 0 & 0 & 0 & 1 & 0 & 0 & 0 \\
                    0 & 1 & 0 & 0 & 0 & 0 & 1 & 1 & 0 \\
                    0 & 0 & 1 & 0 & 0 & 0 & 0 & 0 & 1 \\
                    0 & 0 & 0 & 1 & 0 & 1 & 1 & 0 & 0 \\
                    0 & 0 & 0 & 0 & 1 & 0 & 0 & 1 & 1 \\
                    1 & 0 & 0 & 1 & 0 & 1 & 1 & 0 & 0 \\
                    0 & 1 & 0 & 1 & 0 & 1 & 1 & 1 & 0 \\
                    0 & 1 & 0 & 0 & 1 & 0 & 1 & 1 & 1 \\
                    0 & 0 & 1 & 0 & 1 & 0 & 0 & 1 & 1 \\
                   \end{array}
                   \right] \; .  $$
The Dirac matrix $D=d+d^*$ is 
$$ D = \left[
                   \begin{array}{ccccccccc}
                    0 & 0 & 0 & 0 & 0 & -1 & 0 & 0 & 0 \\
                    0 & 0 & 0 & 0 & 0 & 0 & -1 & -1 & 0 \\
                    0 & 0 & 0 & 0 & 0 & 0 & 0 & 0 & -1 \\
                    0 & 0 & 0 & 0 & 0 & 1 & 1 & 0 & 0 \\
                    0 & 0 & 0 & 0 & 0 & 0 & 0 & 1 & 1 \\
                    -1 & 0 & 0 & 1 & 0 & 0 & 0 & 0 & 0 \\
                    0 & -1 & 0 & 1 & 0 & 0 & 0 & 0 & 0 \\
                    0 & -1 & 0 & 0 & 1 & 0 & 0 & 0 & 0 \\
                    0 & 0 & -1 & 0 & 1 & 0 & 0 & 0 & 0 \\
                   \end{array}
                   \right] $$
made of the incidence matrix $d_0$, (the lower left block) which is 
the gradient mapping $0$-forms to $1$-forms as well as its
adjoint $d_0^*$, (the upper right block) which is the divergence mapping
$1$-forms to $0$-forms. The Hodge Laplacian $H=D^2=(d+d^*)^2$ 
has now a diagonal block structure $H_0 \oplus H_1$, where $H_0$ is the Helmholtz
matrix (a $5 \times 5$ matrix). The $1$-form block is a $4 \times 4$ matrix
$H_1$ which is essentially isospectral to $H_0$. It is an invertible Jacobi
matrix in this case, reflecting that $b_1=0$:
$$ H=              \left[
                   \begin{array}{ccccccccc}
                    1 & 0 & 0 & -1 & 0 & 0 & 0 & 0 & 0 \\
                    0 & 2 & 0 & -1 & -1 & 0 & 0 & 0 & 0 \\
                    0 & 0 & 1 & 0 & -1 & 0 & 0 & 0 & 0 \\
                    -1 & -1 & 0 & 2 & 0 & 0 & 0 & 0 & 0 \\
                    0 & -1 & -1 & 0 & 2 & 0 & 0 & 0 & 0 \\
                    0 & 0 & 0 & 0 & 0 & 2 & 1 & 0 & 0 \\
                    0 & 0 & 0 & 0 & 0 & 1 & 2 & 1 & 0 \\
                    0 & 0 & 0 & 0 & 0 & 0 & 1 & 2 & 1 \\
                    0 & 0 & 0 & 0 & 0 & 0 & 0 & 1 & 2 \\
                   \end{array}
                   \right] \; . $$
The Green's function is 
$$ L^{-1}=g=       \left[
                   \begin{array}{ccccccccc}
                    0 & 0 & 0 & -1 & 0 & 1 & 0 & 0 & 0 \\
                    0 & -1 & 0 & -1 & -1 & 0 & 1 & 1 & 0 \\
                    0 & 0 & 0 & 0 & -1 & 0 & 0 & 0 & 1 \\
                    -1 & -1 & 0 & -1 & 0 & 1 & 1 & 0 & 0 \\
                    0 & -1 & -1 & 0 & -1 & 0 & 0 & 1 & 1 \\
                    1 & 0 & 0 & 1 & 0 & -1 & 0 & 0 & 0 \\
                    0 & 1 & 0 & 1 & 0 & 0 & -1 & 0 & 0 \\
                    0 & 1 & 0 & 0 & 1 & 0 & 0 & -1 & 0 \\
                    0 & 0 & 1 & 0 & 1 & 0 & 0 & 0 & -1 \\
                   \end{array}
                   \right]  \; . $$
The Hydrogen operator is the sign-less Hodge matrix. It is
a non-negative matrix:
$$ L-g =           \left[
                   \begin{array}{ccccccccc}
                    1 & 0 & 0 & 1 & 0 & 0 & 0 & 0 & 0 \\
                    0 & 2 & 0 & 1 & 1 & 0 & 0 & 0 & 0 \\
                    0 & 0 & 1 & 0 & 1 & 0 & 0 & 0 & 0 \\
                    1 & 1 & 0 & 2 & 0 & 0 & 0 & 0 & 0 \\
                    0 & 1 & 1 & 0 & 2 & 0 & 0 & 0 & 0 \\
                    0 & 0 & 0 & 0 & 0 & 2 & 1 & 0 & 0 \\
                    0 & 0 & 0 & 0 & 0 & 1 & 2 & 1 & 0 \\
                    0 & 0 & 0 & 0 & 0 & 0 & 1 & 2 & 1 \\
                    0 & 0 & 0 & 0 & 0 & 0 & 0 & 1 & 2 \\
                   \end{array}
                   \right]  \; . $$
The eigenvalues of $H$ are $\sigma(H)=\left\{ \right.$ $\frac{1}{2} \left(5+\sqrt{5}\right)$, $\frac{1}{2}
    \left(5+\sqrt{5}\right)$, $\frac{1}{2} \left(3+\sqrt{5}\right)$, $\frac{1}{2}
    \left(3+\sqrt{5}\right)$, $\frac{1}{2} \left(5-\sqrt{5}\right)$, $\frac{1}{2}
    \left(5-\sqrt{5}\right)$, $\frac{1}{2} \left(3-\sqrt{5}\right)$, $\frac{1}{2}
    \left(3-\sqrt{5}\right)$, $0 \left. \right\}$.
The eigenvalues of $L$ are $\sigma(L) = 
\{3.87603$, $2.9563$, $1.90649$, $1.20906$, $1$, $-0.827091$, $-0.524524$, $-0.338261$, $-0.257996$ $\}$.
The eigenvalues of $L^{-1}$ are $\sigma(g) = 
 \{-3.87603$, $-2.9563$, $-1.90649$, $-1.20906$, $1$, $0.827091$, $0.524524$, $0.338261$, $0.257996$ $\}$
in accordance to the Zeta function functional equation assuring that $L^2$ and $L^{-2}$ have the same 
eigenvalues if the complex is one-dimensional \cite{DyadicRiemann}.

\paragraph{}
Using a coordinate change with diagonal matrix $U$
having entries $\omega(x)$ for ${\rm dim}_G(x)=0$ and $1$ for 
${\rm dim}_G(x)=1$, the matrix $H^+=U H U^T$ is the sign-less Hodge Laplacian.
Now, $L-L^{-1} = U H U^T  = H^+$. 
We have implemented the matrices explicitly using a computer algebra system and included
the code at the end. There are many puzzles which remain: we
have no idea yet for example how to fix a relation between $L$ and $H$ in the higher dimensional case.
We believe that there should be a deformation of $L$ which still makes this happen as we have seen 
examples, where a change of $L$ works. For a triangle complex $K_3$ for example, we just
have to change the interaction energy between the $2$-dimensional simplex and the others
and still get $H=L-L^{-1}$. Maybe, in general, a small tuning suffices to achieve a connection of $H$
with a non-negative $0-1$ matrix $L$ for which the spectral analysis is easier. 

\paragraph{}
As explained below, one can get intuition from physics. The matrix 
entries can be seen as manifestations of energy potentials. There are indications in the form
of examples which suggest that we can change $L$ to still have a Hydrogen formula. 
This leads to gauge fields. Already the conjugation $H \to U H U^T$ is a gauge
change even so trivial. It corresponds to a gauge field. The hope is that
the inclusion of more general gauge fields (which changes the spectrum of $L$) would allow
to save the algebraic relation between $L$ and $H$ also in higher dimensions. 

\begin{coro}
For a Barycentric refined graph, the sign-less Hodge matrix $H$ is a 
non-negative matrix which satisfies $H=L-L^{-1}$ and $H^2+2=L^2+L^{-2}$.
\end{coro}

\paragraph{}
A matrix $A$ is called reducible if there is a basis in which it can be written 
$A=A_1 \oplus A_2$. If no such decomposition is possible, 
the matrix is called irreducible. For a non-negative matrix, 
irreducibility is equivalent to the statement that
there exists $k$ such that $A^k$ is a positive matrix, 
meaning that all entries $A^k(x,y)$ are positive. 
Both $H$ and $L$ are non-negative matrices in a suitable basis. 
If the graph is connected and is not zero-dimensional, then 
$H$ is reducible but $L$ is irreducible. We can also use that in 
one dimensions, the spectrum of $L^2$ is the same than the spectrum
of $L^{-2}$. 

\begin{coro}
For a connected Barycentric refined graph, the maximal eigenvalue
of the Kirchhoff matrix $H_0$ has multiplicity $1$. 
\end{coro}

\paragraph{}
{\bf Example.} If $\Gamma=C_n$ is a circular graph with even $n$, 
then $\lambda_k=4 \sin^2(\pi k/n)$ are the eigenvalues of $H=L-L^{-1}$ 
and $2+16 \sin^4(\pi k/n)$ are the eigenvalues of $K=H^2+2=L^2+L^{-2}$. 
This example was the first time we have seen the relation $H=L-L^{-1}$. 
It was essential to find an explicit Zeta function of in the Barycentric limit
\cite{DyadicRiemann}.

\section{The proof of the theorem} 

\paragraph{}
In order to prove the result, we need to know the matrix entries of 
the Green's function $g=L^{-1}$. We know already the diagonal entries
$g(x,x) = 1-\chi(S(x))$. 
This means that for a zero-dimensional simplex $x$, we have $g(x,x)=1-d(x)$,
where $d(x)$ is the vertex degree and $g(x,x)=2$ if $x$ is a one-dimensional simplex. 
Because $L(x,x)=1$, we have $L(x,x)-g(x,x) = d(x)$ on the zero-dimensional sector
and $L(x,x)-g(x,x)=2$ for the one-dimensional sector. This matches the matrix entries
of $H$ in the diagonal. In order to prove the result we have to establish: \\

\begin{lemma}
a) $g(x,y)=0$ if ${\rm dim}(x) \neq {\rm dim}(y)$ and $x \cap y = \emptyset$. \\
b) $g(x,y)=1$ if ${\rm dim}(x) \neq {\rm dim}(y)$ and $x \subset y$ or $y \subset x$. \\
c) $g(x,y)=-1$ if ${\rm dim}(x)={\rm dim}(y)=0$ and $(x,y) \in E$. \\
d) $g(x,y)=0$ if ${\rm dim}(x)={\rm dim}(y)=1$. 
\end{lemma}
\begin{proof} 
We can use this data to build each column vector of $g$. Now just compute
the dot product of a $y$ row vector $v$ of $L$ with a $x$ column vector $w$ of $g$ 
and compute $v \cdot w$. If $x=y$, this is $1$ as there is a hit $1-\chi(S(x))$
and then there are $S(x)$ terms $-1$. For $x \neq y$, then the dot product has
only two terms, one being $1$, the other $-1$. Here are a bit more details even so
the general case will make this obsolete: \\
As $G$ is a Barycentric refinement, its vertices are 2-colorable.
Any coloring with coloring $0,1$ as well as an orientation of the edges defines a basis.
The alternating sign change of the basis assures that
$H_0^+ = B+A$ which is the sign-less Kirchhoff matrix. It is isospectral to $H_0=B-A$.
We will show that $L-L^{-1} = B+A$. \\
a) First the diagonal: since $H=L-g$, where $g$ is the Green's
operator, we know all the entries of $g$. In the diagonal we have $g(x,x) = 1-\chi(S(x))$
$= 1-{\rm deg}(x)$. As $L(x,x)=1$, we have $(L-L^{-1}(x,x)=H(x,x) = {\rm deg}(x)$. Note
that this works also on the 1-form sector as every edge has exactly two neighbors
and therefore $(L-L^{-1})(x,x) =1-(1-\chi(S(x)) = \chi(S(x))=2$ for an edge. \\
b) Now the mixed dimension part:
assume $x \subset y$ where $x$ is zero dimensional and $y$ is one dimensional. Then
we know $L(x,y)=1$ and $L^{-1}(x,y)=\omega(x)=1$. This means that $(L-L^{-1})(x,y)=0$
so that $L - L^{-1}$ has a block structure. \\
c) Now we look at the case where $x,y$ are both zero dimensional. Then $L(x,y)=0$
and $L^{-1}(x,y) = -1$. This agrees with $H^+(x,y) = B+A(x,y)$. \\
d) Finally look at the case where $x,y$ are both $1$-dimensional. Then $L(x,y)=1$
if $x \cap y$ is not empty and $L(x,y)=0$ else. We can use
$L^{-1}(x,y)=0$ to get $(L-L^{-1})(x,y)=1$. Also, if $x,y$ do not intersect,
then $L^{-1}(x,y)=0$.
\end{proof}

\paragraph{}
The full generalization uses the "star" ${\rm St}(x)$ of $x$,
which is the set of all $y \in G$ if $x \subset y$. It is a collection of simplices, but not a simplicial complex 
in general. It defines a graph $S^+(x)$ in the Barycentric refinement $G_1$.
There is a subtlety: while we know that for a simplicial complex $G$,
the Euler characteristic of $G$ and its Barycentric refinement $G_1$ are the same, this is not true for 
sets of simplices which are not simplicial complexes. Take $A=\{ \{1\}, \{1,2\} \}$ which is not a simplicial complex
but which has Euler characteristic $\chi(A) = \sum_{x \in A} \omega(x)$ with $\omega(x)=(-1)^{{\rm dim}(x)-1} = 1-1=0$.
The Barycentric refinement $A_1$ of $A$ is now the complete graph $K_2$
which has Euler characteristic $1$. It is the fact that the star is not a simplicial complex which 
requires us to compute in $G$ and does not allow us 
not escape to its Barycentric refinement, which is a graph.

\paragraph{}
The following "Green star formula" is the ultimate answer 
about the Green function entries.  

\begin{propo}[Green Star formula]
$$ g(x,y) = \omega(x) \omega(y) \chi( {\rm St}(x) \cap {\rm St}(y) ) $$
\end{propo}
\begin{proof}
(Sketch) We know by Cramer that $g(x,y) =  {\rm adj}(L)(x,y)/{\rm det}(L)$, 
where ${\rm adj}(L)$ is the matrix $L$ with row $x$ and column $y$ deleted. 
Now proceed by induction in the same way as for the unimodularity theorem. 
For proving the formula for a pair $x,y$, consider an other maximal simplex $z$ away from
$x$ and $y$ (which is possible if we don't deal with a complete graph), 
then use the multiplicative Poincar\'e-Hopf formula for the 
change of the determinant: both sides are multiplied by $1-\chi(S(z))$. See
\cite{Unimodularity}. 
\end{proof} 

\paragraph{}
The Green star formula gives the inverse in a concrete way. 
One can also write the matrix multiplication $L g = 1$ and verify each entry: 
$$ \sum_{z} L(x,z)   \omega(z) \chi( {\rm St}(z) \cap {\rm St}(y) )  = 0 $$
if $x \neq y$ and
$$ \sum_{z} L(x,z)   \omega(z) \chi( {\rm St}(z) \cap {\rm St}(x) )  = \omega(x) \; .  $$
Using the notation $z \sim x$ if $x \cap z$ intersect: it means for $x \neq y$
$$ \sum_{z \sim x}  \omega(z) \chi( {\rm St}(z) \cap {\rm St}(y) )  = 0 $$
and 
$$ \sum_{z \sim x}  \omega(z) \chi( {\rm St}(z) \cap {\rm St}(x) )  = \omega(x) \; .  $$
These are both local Gauss-Bonnet statements similar as in \cite{Helmholtz}. The Green
star formula is equivalent to these two statements about stars in simplicial complexes. 

\paragraph{}
{\rm Remarks.} \\
{\bf 1)} We have in particular $g(x,x)=\chi({\rm St}(x))$
which means the self-interaction energy of a simplex is the Euler characteristic of its star.
{\bf 2)} If we look at the dual star $W^-(x)={\rm St}^-(x)$ of $x$, which is the set of all $y \in G$
with $y \subset x$, then this is a complete simplicial complex with Euler characteristic $1$. We 
can now write
$$ L(x,y) = \chi( W^-(x) \cap W^-(y) ) \; . $$
We see from this that the matrix $L$ refers to the inside stable part of the simplices
while the inverse matrix $g$ refers to the outside unstable part of the simplices. 

\paragraph{}
The connection matrix $L$ is conjugated to 
$$ M(x,y) = \omega(x) \omega(y) \chi( W^-(x) \cap W^-(y) ) \; . $$
The inverse $g$ is conjugated via the diagonal matrix ${\rm Diag}(\omega(x))$ to
$$ h(x,y) = \chi( W^+(x) \cap W^+(y) )  \; . $$
We see an obvious duality. Mending the two pictures requires to 
go into the complex. Define the diagonal matrix $U$ which has the diagonal 
entries $U(x,x) = \sqrt{\omega(x)}$. Now we can look at
$$ Y = U (L-g) U $$
While $L-g$ is isospectral to the Hodge Laplacian $H=(d+d^*)^2$, the 
turned operator $Y$ is of the form $H_0 \oplus (-H_1)$. Now, paired with 
the energy theorem $\sum_x \sum_y g(x,y) = \chi(G)$, 
we have a relation with the Wu characteristic 
$$ \omega(G) = \sum_{x,y} L(x,y) \omega(x) \omega(y) $$
which is the total energy of the operator $M$. 

\begin{propo}
For any $1$-dimensional complex which is a Barycentric refinement, we have
$\sum_x \sum_y Y(x,y)  = \chi(G) - \omega(G)$. 
\end{propo}

\paragraph{}
Now this is interesting, as the energy is still a combinatorial invariant, a quantity 
which does not change if we make a Barycentric refinement. We have actually proven 
in \cite{valuation}, see also \cite{DehnSommerville} that for geometric complexes with
boundary, $\chi(G)-\omega(G) = \omega(\delta G)$. If we interpret the curvature for
$\chi(G)-\omega(G)$ as an energy of $G$, we see that it is located on the boundary of a complex
and that the total energy is zero in the geometric case. For a closed 
circular graph for example, the total energy of $Y$ is zero. In higher dimensions, the gauge fields
have to be added differently and it is still unclear whether one can deform $L$ to mend the
Hydrogen formula. If it is possible, then most likely through a variational mechanism which
by wishful thinking should relate to some kind of radiation. 

\begin{propo}
In general, for any complex $G$, the total energy of $M-g$ is $\chi(G)-\omega(G)$. 
This total energy is zero for geometric graphs without boundary. The energy curvature
is supported on the boundary of a geometric space with boundary. 
\end{propo}

\paragraph{}
This is not that unfamiliar if we compare a simplicial complex with a space or space time
manifold in physics. These manifolds naturally have boundaries as event horizons of
singularities. Now, as Hawking famously first pointed out, these boundaries
radiate. The analogy is certainly far fetched as what we deal here with relatively basic 
combinatorial geometry of finite set of sets. Still, it is  a mathematical fact that if we define
energy of such a geometry as $\chi(G)-\omega(G)$, where $\chi$ is the Euler characteristic and $\omega$
is the Wu characteristic, then due to Dehn-Sommerville, in the interior of Euclidean
like parts of space, the energy density (curvature) is zero and
all the energy density is at the boundary or located at topological defects of space.
History cautions to speculate as the molecular vortex picture debacle reminds.
But fundamental questions about the nature of space and time has always
motivated mathematics. Here, we deal with remarkable mathematical theorems like
the energy theorem which assures that the sum over all interaction 
energies of simplices in a simplicial complex is the Euler characteristic 
of the simplicial complex. And this was certainly motivated by physics of the Laplacian in 
Euclidean space. 

\paragraph{}
Let us explain why for any simplicial complex $G$ and any simplex $x \in G$
the "star formula" $1-\chi(S(x)) = \chi({\rm St}(x))$ holds. 
To see the "star formula", we write the unit sphere $S(x)$ as the Zykov join of its stable and unstable part 
$S(x)=S^+(x) + S^-(x)$ and use that the "genus" $i(A)= 1-\chi(A)$ is multiplicative
$1-\chi(S(x)) = (1-\chi(S^+(x))) (1-\chi(S^-(x)))$. Now, since the stable sphere $S^-(x)$ 
is the boundary of a simplicial complex, we have 
$1-\chi(S^-(x)) = \omega(x)=(-1)^{{\rm dim}(x)}$. The statement
$(1-\chi(S^+(x))) \omega(x) = \chi(St(x))$ is true because every simplex in ${\rm St}(x)$
is bijectively related to a simplex ${\rm St}(x) \setminus x$ in $S(x)$. A vertex $v$ in $S(x)$
corresponds to a simplex $x \cup v$ in $St(x)$. In some sense, collapsing the simplex $x$
in the star ${\rm St}(x)$ to a point and removing that point gives the stable sphere $S^+(x)$. 

\paragraph{}
The fact that for a bipartite graph, $H_0^+=B+A$ and 
the Kirchhoff matrix $H_0=B-A$ are unitarily equivalent
appears in Proposition 2.2 of \cite{GroneMerrisSunder1}. 
The result more generally holds for completely positive graphs, as these are the
graphs which have no odd cycles of length larger than $4$ \cite{Berman1993}.
It already does not apply for the Barycentric refinemd triangular graph $(K_3)_1$, 
where the eigenvalues of $H_0$ are $\{7, 5, 4, 4, 2, 2, 0\}$ and the eigenvalues of 
$H_0^+$ are $\{8, 4, 4, 3, 2, 2, 1 \}$.

\paragraph{}
As an application we can compute the eigenvalues of the connection matrices of
classes of $1$-dimensional operators like circular graphs. The result also sheds light on the 
eigenvalue structure. Any integer eigenvalue for $L$ different from $1,-1$ leads
to non-integer eigenvalues of $H$. More applications are likely to follow. 

\section{Hearing the cohomology}

\paragraph{}
For a $1$-dimensional complex $G$, the Hodge operator $H=D^2$ decomposes into two 
blocks $H_0 \oplus H_1$. The Betti numbers are then $b_k={\rm dim}({\rm ker}(L_k)$ for $k=0,1$. 
The $b_0$ counts the number of connectivity components, the number $b_1$ counts
the number of generators of the fundamental group. We in general can not hear the cohomology 
of a complex, when listening to $L$. The example given is even one-dimensional 
\cite{HearingEulerCharacteristic}. 

\paragraph{}
We assume the complex $\Gamma$ to be a Barycentric refinement of $G$. This implies chromatic number
2 for $\Gamma$ and that $\Gamma$ is bipartite. Lets look first at the eigenvalue $1$: 

\begin{lemma}
For every connectivity component, we have an eigenvalue $1$ of $L$. The eigenvector is 
a $\{-1,1\}$-coloring supported on vertices of $\Gamma$ which 
were zero dimensional in $G$. Every connected component has exactly one eigenvalue $1$. 
\end{lemma}

\begin{proof}
We only have to show that every eigenvector $f$ to an eigenvalue $1$ is supported
on vertices of $\Gamma$ which were zero dimensional in $G$. 
This can be done by induction on the number of $1$-dimensional vertices, (vertices
in $\Gamma$ which were $1$-dimensional in $G$). Lets prove more generally that any 
eigenvector to the eigenvalue $1$ is supported on the $0$-dimensional part
of the complex, where it is necessarily a coloring. 
From the fact that $\sum_{h \sim v, h \in E}  f(h) = 0$ we get 
$\sum_{h \in E} f(h)=0$. We can now use induction. 
Lets call a vertex with vertex degree $1$ a "leaf".
An eigenvalue $1$ corresponds to an eigenvalue $0$ of the adjacency matrix
of the connection graph. For every vertex $v$, the average of all values on 
edges connected to $v$ is zero. Assume there is an edge $e$ with a leaf attached.
Then $f(e)=0$. So, we can apply induction and remove the leaf. Without
any leaf, the original complex $G$ must have been a closed loop. Assume that 
$f(e)>0$ for some $e$. When looking at vertices we see $f(e)$ changes sign along 
the edges of the loop and that the two neighbors have the same sign
and add up to zero.
\end{proof} 

\paragraph{}
Now, we look at the eigenvalue $-1$: 

\begin{lemma}
Every homotopically non-trivial closed cycle leads to an eigenvalue $-1$. The eigenvector is 
a $\{-1,1\}$ coloring of the edges of the cycle and supported on edges. 
A basis of the eigenspace of $\lambda_1$ corresponds to a generating
set of the fundamental group. 
\end{lemma}

\begin{proof}
Every cycle leads to an eigenvector: just put alternating
values $1,-1$ on the edges of a cycle and put $0$ everywhere else. 
The fact that every eigenvector can be traced back to a closed path is
a consequence of the Hurwicz theorem relating the fundamental group with the
first homology group $H^1$. The Hurwicz homomorphism is explicit for $H$: 
take a closed path and build from it a function $f$ on edges telling how many times
an edge has been traversed incorporating the direction. Now apply the heat flow
$\exp(-t H)$ on this function. As the Hodge matrix $H$ has only nonnegative eigenvalues,
the positive eigenvalue part will die out and the limit will be 
located on the kernel of $H$, which gives a representative of the 
cohomology group $H^1$. It can also be seen from the 
relation that $L-L^{-1}$ is similar to $H$ as we have already taken care
of the eigenvalues $0$ of $H$ which come from eigenvalues $1$ of $L$. 
and $L-L^{-1}$ has eigenvectors with the same support than $H$ as the conjugation
is done by a diagonal matrix. 
\end{proof} 

\begin{figure}
\scalebox{0.42}{\includegraphics{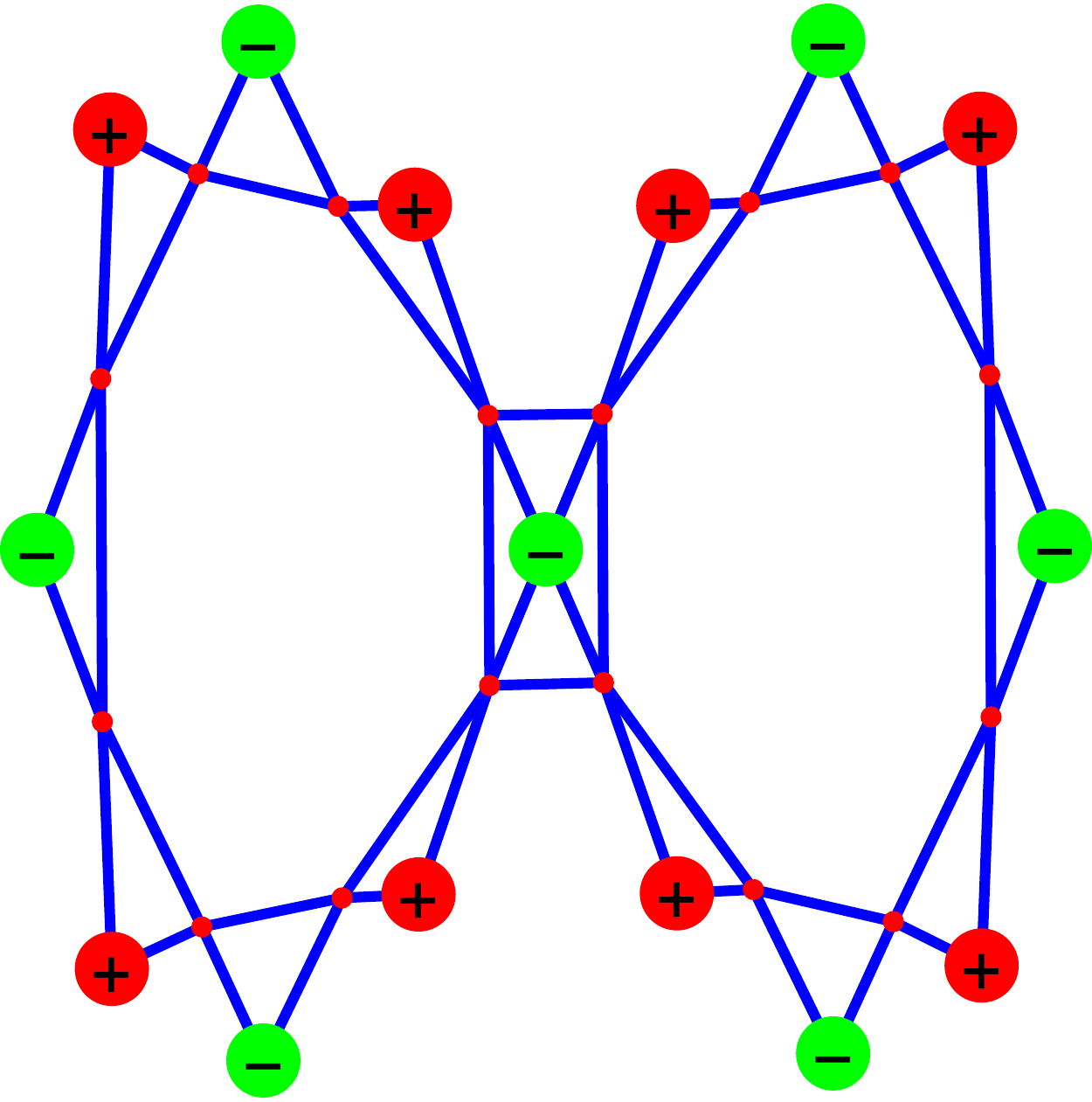}} \\
\scalebox{0.42}{\includegraphics{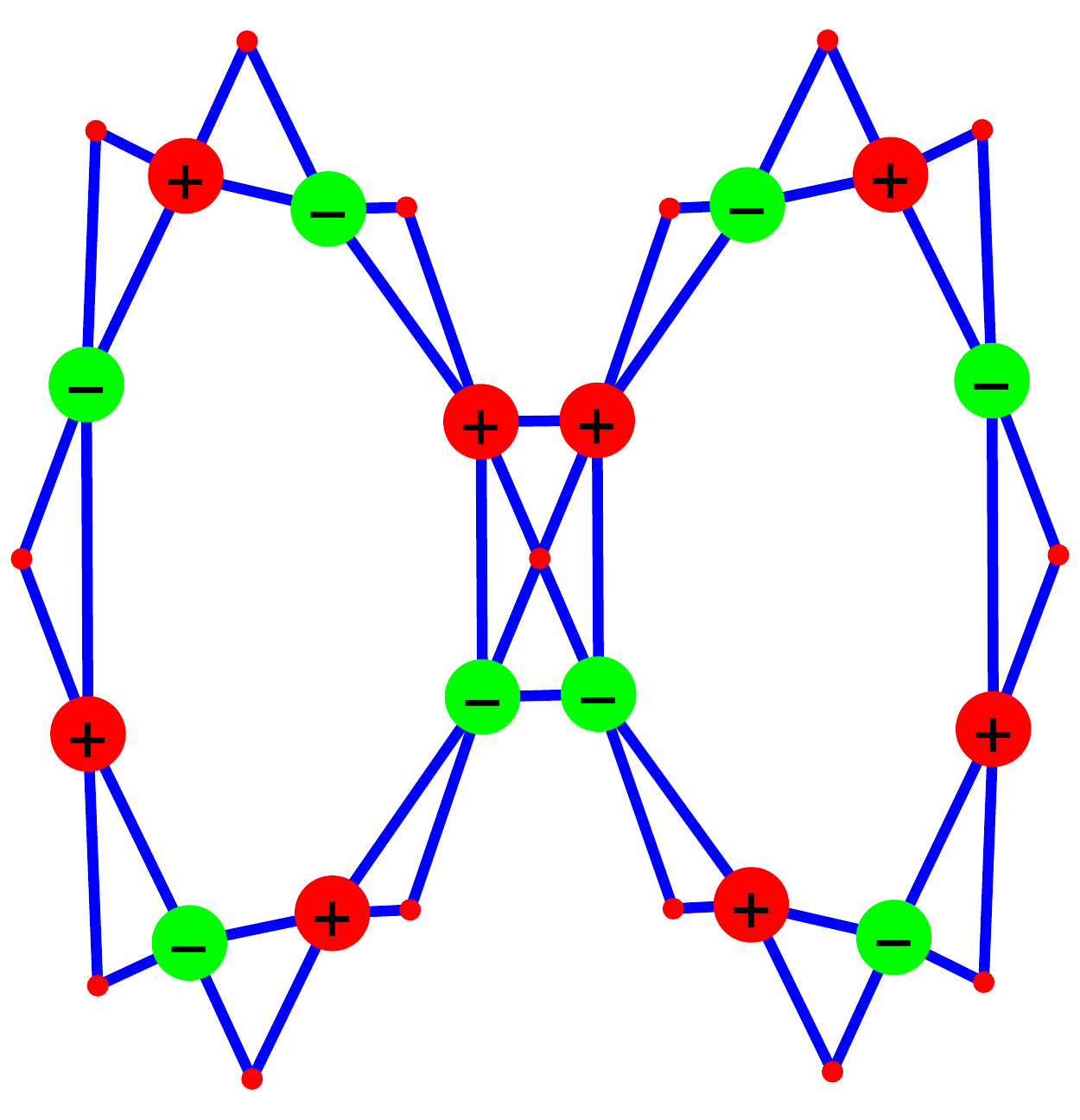}}
\scalebox{0.42}{\includegraphics{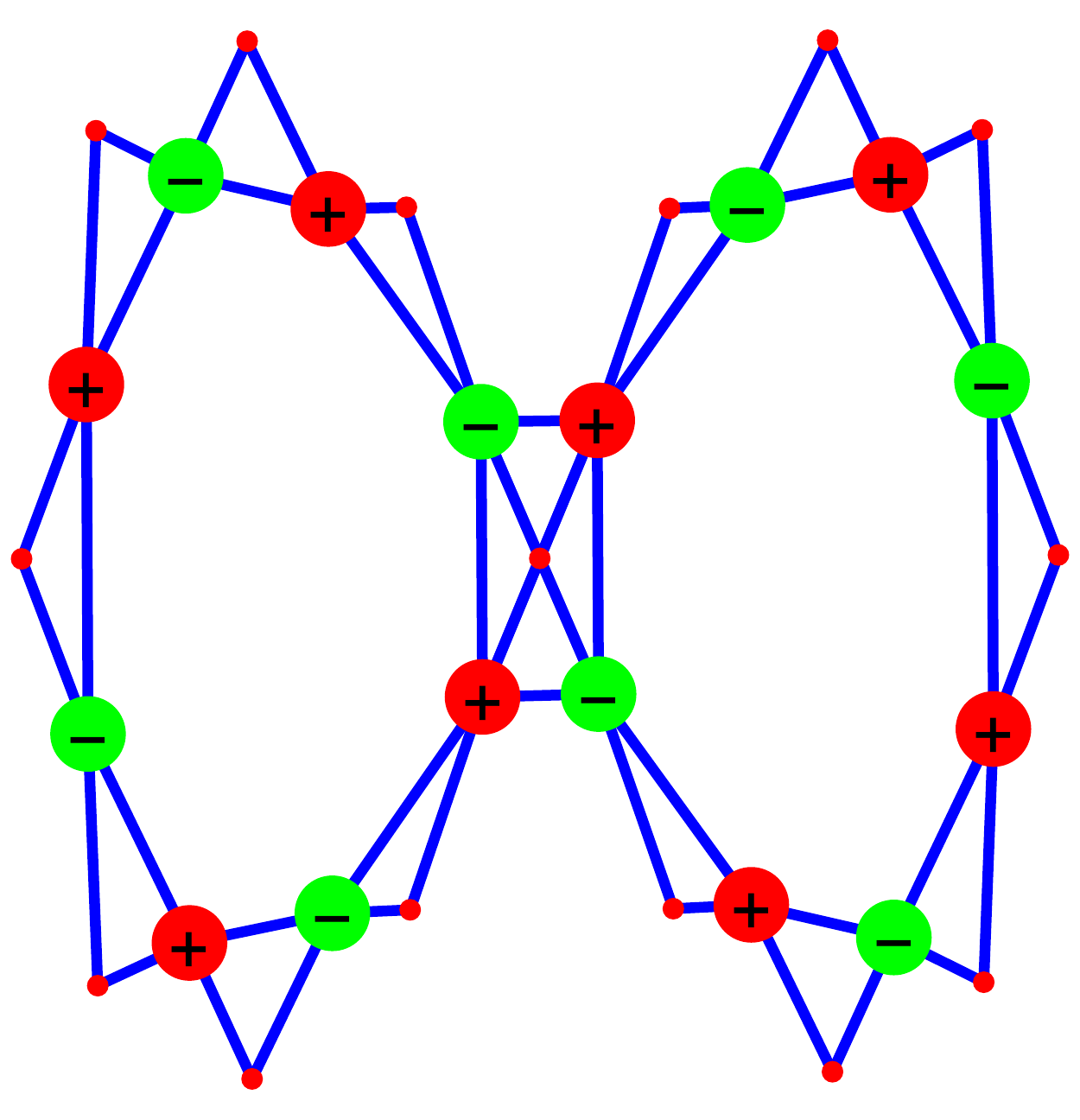}}
\caption{
\label{eigenvectors}
The connection graph $\Gamma'$ of the figure 8 graph $G$. It is larger than
the Barycentric refinement $\Gamma$ as in the connection graph also edges are 
connected and has triangles. We see first the eigenvector to the eigenvalue $1$. It is supported on 
the zero-dimensional parts of the vertex set of $\Gamma$ (the vertices
which were $0$-dimensional in $G$).
Then we see the two eigenvectors to the two eigenvalues $-1$. 
They are supported on one-dimensional parts (vertices of $\Gamma$ which were
$1$-dimensional in $G$). }
\end{figure}

\paragraph{}
Can we see from the eigenvalues of $L$, which ones belong to $1$-forms
and which one belong to $0$ forms?  
The positive eigenvalues of $L$ are the ones from the $0$-forms 
and the negative eigenvalues of $L$ belong to the $1$-forms. 
For higher dimensional complexes, the eigenvalues of $M=L-L^{-1}$ can take both 
values. In the one dimensional case, the eigenvalues are always non-negative. 
We can also describe every point of $H(x,y)$. But in $M(x,y)$, we have connections
between vertices and edges, while in $H(x,y)$ we have connections between vertices
and vertices, and edges with edges. 

\paragraph{}
If $G=G_1 \times G_2$ is the product of two $1$-dimensional 
complexes which are Barycentric refinements, then the cohomology
of $G$ is determined by the cohomology of $G_i$. Assume
we know the eigenvalues of $G$, we can from the multiplicities
get the eigenvalues of $G_i$ and so the eigenvalues of the Hodge
operators $H_i$ and from this the eigenvalues of the Hodge 
operator $H$ of $G$. Can we do that in general for an arbitrary
number of products $G_i$? 

\paragraph{}
One should probably first focus on the $b_2$ case and
try to hear the second cohomology $b_2$ from the 
spectrum of $L$. We made some experiments with random complexes
and counted the number of different eigenvalues of $H$ and compared this
with $b_2$. We tried to correlated $b_2$ with spectral data like 
the fraction $\sigma(L)/n$, where $\sigma(G)$ counts the number 
of different eigenvalues of the $n \times n$ matrix $L$. Also no
relation between the factorization of the characteristic polynomial
and $b_2$ has been found yet. But there are other correlations still to
be tried out, like relations between moments ${\rm tr}(L^k)$ and
cohomology or zeta function values with cohomology. 

\section{Mathematica content}

\paragraph{}
Here is an illustration of the Hydrogen formula $H \sim L-L^{-1}$. 
We take a random graph, refine it to get a one-dimensional complex
with chromatic number $2$, then build both $L,H$ and the conjugation
diagonal matrix $R$.

\begin{tiny}
\lstset{language=Mathematica} \lstset{frameround=fttt}
\begin{lstlisting}[frame=single]
sort[x_]:=Sort[{x[[1]],x[[2]]}]; v=5;e=10; (* size of the graph    *)
Gra=RandomGraph[{v,e}]; bracket[x_]:={x};  f={v,e}; q=v+e;
Q=Union[Map[sort,EdgeList[Gra]],Map[bracket,VertexList[Gra]]]; 
G=Map[bracket,Range[q]];   (* start building Barycentric refinement*)
Do[If[SubsetQ[Q[[k]],Q[[l]]] && k!=l,G=Append[G,{k,l}]],{k,q},{l,q}];
G=Union[Map[Sort,G]]; v=q; n=Length[G];  
(* have built now a random Barycentric refined 1-dim complex G  *)
Orient[a_,b_]:=Module[{z,c,k=Length[a],l=Length[b]},
  If[SubsetQ[a,b] && (k==l+1),z=Complement[a,b][[1]];
  c=Prepend[b,z]; Signature[a]*Signature[c],0]];
d=Table[0,{n},{n}]; d=Table[Orient[G[[i]],G[[j]]],{i,n},{j,n}]; 
Dirac=d+Transpose[d]; H=Dirac.Dirac;   (* Hodge Laplacian is built *)
L=Table[If[DisjointQ[G[[k]],G[[l]]],0,1],{k,n},{l,n}]; 
R=DiagonalMatrix[Table[If[k<=v,(-1)^Length[Q[[k]]],1],{k,n}]];
Total[Flatten[Abs[R.(L-Inverse[L]).R - H]]]
\end{lstlisting}
\end{tiny}

\paragraph{}
We illustrate now the Green star formula
$$ g(x,y) = \omega(x) \omega(y) \chi( {\rm St}(x) \cap {\rm St}(y) )  \; . $$
The code computes for an arbitrary simplicial complex the Green's function 
entires of the inverse matrix $g=L^{-1}$ 
of the connection matrix $L$ in terms of the stars ${\rm St}(x)=W^{+}(x)$ and 
${\rm St}(y)=W^+(y))$. The dimension functional $x \to {\rm dim}(x)$ on $G$ defines a locally 
injective function which can be seen as a Morse function. The gradient 
flow of this functional has stable and unstable manifolds $W^+(x)$ and $W^-(x)$. 
Speaking in the language of hyperbolic dynamics, the homoclinic
tangle of this "Morse-Smale" system produces the Green functions.
Every simplex is a critical point and the definition of Euler characteristic 
is a special case of the Morse inequality. Indeed, the component $v_k$ of the 
$f$-vector of $G$ counts the number of critical points having Morse index $k$
(The Morse index of a simplex is the dimension of the stable manifold, which 
is here the $1$ plus the dimension of the sphere $W^-(x) \cap S(x)$).

\paragraph{}
The index $\omega(x) = \chi(W^+(x) \cap W^-(x))$ is related to a homoclinic
point, the matrix entries $L(x,y) = \chi(W^-(x) \cap W^-(y))$ 
of the connection matrix and the matrix entries $\chi(W^+(x) \cap W^+(y))$ 
form the matrix entries of a matrix $R g R^{-1}$ conjugated to the Green's
function, where $R$ is the diagonal matrix with $\omega(x)$ entries. 
The Wu matrix $M = R L R^{-1}$ is conjugated to $L$ and its matrix entries
add up to the Wu characteristic $\omega(G)=\sum_{x \sim y} \omega(x) \omega(y)$. 

\begin{tiny}
\lstset{language=Mathematica} \lstset{frameround=fttt}
\begin{lstlisting}[frame=single]
Generate[A_]:=Delete[Union[Sort[Flatten[Map[Subsets,A],1]]],1]
Ra[n_,m_]:=Module[{A={},X=Range[n],k},Do[k:=1+Random[Integer,n-1];
  A=Append[A,Union[RandomChoice[X,k]]],{m}];Generate[A]];
G=Ra[7,15];n=Length[G];SQ=SubsetQ; OmegaComplex[x_]:=-(-1)^Length[x];
EulerChiComplex[GG_]:=Total[Map[OmegaComplex,GG]];
S[x_]:=Module[{u={}},Do[v=G[[k]];If[SQ[v,x],u=Append[u,v]],{k,n}];u];
K=Table[OmegaComplex[G[[k]]] OmegaComplex[G[[l]]],{k,n},{l,n}];
H=Table[Intersection[S[G[[k]]],S[G[[l]]]],{k,n},{l,n}];
h = Table[K[[k, l]] EulerChiComplex[H[[k, l]]], {k, n}, {l, n}];
L=Table[If[DisjointQ[G[[k]],G[[l]]],0,1],{k,n},{l,n}]; 
h.L==IdentityMatrix[n]
\end{lstlisting}
\end{tiny}

\pagebreak

\bibliographystyle{plain}

\end{document}